\documentclass[journal,twoside,web]{ieeecolor}
\usepackage{lcsys}
\usepackage{cite}
\usepackage{amsmath,amssymb,amsfonts}
\usepackage{algorithmic}
\usepackage{graphicx}
\usepackage{textcomp}
\usepackage{url}
\usepackage{caption}
\usepackage{subcaption}
\usepackage{hyperref}
\def\BibTeX{{\rm B\kern-.05em{\sc i\kern-.025em b}\kern-.08em
    T\kern-.1667em\lower.7ex\hbox{E}\kern-.125emX}}
\newtheorem{definition}{Definition}
\newtheorem{assumption}{Assumption}
\newtheorem{problem}{\textbf{Problem}}

\newtheorem{proposition}{\textbf{Proposition}}

\newcommand{\rr}{\mathop{{\rm I}\mskip-4.0mu{\rm R}}\nolimits}
\newcommand{\Z}{\mathop{{\rm Z}\mskip-7.0mu{\rm Z}}\nolimits}

\newcommand{\Enc}[1]{\texttt{E}[#1]}
\newcommand{\Dec}[1]{\texttt{D}[#1]}

\begin{document}

\title{A Verifiable Computing Scheme for Encrypted Control Systems}
\author{Francesca Stabile,  Walter Lucia,  Amr Youssef, Giuseppe Franz\`e
\thanks{The work was supported by Mitacs under the grant IT34199.}
	\thanks{Francesca Stabile and Giuseppe Franz\`e are with DIMEG, Universit\`{a}  della Calabria, Via Pietro Bucci, Cubo 42-C,  Rende (CS), 87036, Italy, {\tt \small stabilefrancesca99@outlook.com}, {\tt \small giuseppe.franze@unical.it.} 
} 
\thanks{Walter Lucia and Amr Youssef are with the  Concordia Institute for Information Systems Engineering (CIISE), Concordia University, Montreal,  Canada, {\tt\small walter.lucia@concordia.ca}, {\tt\small youssef@ciise.concordia.ca}. 
}
}

\maketitle

\thispagestyle{empty}

\begin{abstract}
The proliferation of cloud computing technologies has paved the way for deploying networked encrypted control systems, offering high performance, remote accessibility and privacy. However, in scenarios where the control algorithms  run on third-party cloud service providers, the control's logic might be changed by a malicious agent on the cloud. Consequently,  it is imperative to verify the correctness of the control signals received from the cloud.  Traditional verification methods, like zero-knowledge proof techniques, are computationally demanding in both proof generation and verification, may require several rounds of interactions between the prover and verifier and, consequently, are inapplicable in real-time control system applications. In this paper, we present a novel computationally inexpensive verifiable computing solution inspired by the probabilistic cut-and-choose approach. The proposed scheme allows the plant's actuator to validate the computations accomplished by the encrypted cloud-based networked controller without compromising the control scheme’s performance. We showcase the effectiveness and real-time applicability of the proposed verifiable computation scheme using a remotely controlled Khepera-IV differential-drive robot.
\end{abstract}


\section{Introduction}\label{sec:introduction}
\IEEEPARstart{C}{Cloud}-based  control systems \cite{xia2022brief} exploit the well-established advantages of cloud and edge computing, such as increased reliability, easier scalability and reduced IT expenses to enhance the performance of various Cyber-Physical Systems (CPSs). However, these cloud-based systems hinge on communication links for exchanging measurement data and control signals between the plant and the cloud-based controller. Thus, the early research in the field of CPS security has been mainly centered on safeguarding these systems from network attackers with access to the measurement and/or control channels  \cite{teixeira2012attack}. Various control theory based approaches were proposed to mitigate the vulnerability to false data injection attacks \cite{giraldo2018survey}. 
To ensure the privacy of the computation at the cloud controller, homomorphic encryption (HE) schemes were proposed to implement polynomial control laws  on encrypted  data \cite{kim2016encrypting},  \cite{schluter2023brief}, \cite{darup2021encrypted}.  On the other hand, given the safety-critical nature of such systems, it is also necessary to provide the physical plant/smart actuator with an efficient mechanism to validate the correctness of the computations conducted by the cloud-based controller on encrypted data. Verifiable computing   aims not only to obtain the control input computed by the cloud and a proof of its correctness, but also to enable the plant/smart actuator to verify this proof with significantly less computational effort than calculating the control input from scratch. The simplest, yet least efficient, approach for verifiable computation is replication. In this scenario, the verifier (pant/actuator) utilizes multiple, presumably non-colluding cloud-based controllers to perform identical computation tasks. Upon receiving a minimum consensus of results from these multiple controllers, the verifier assumes those results are correct. 
An alternative method to ensure the integrity of the computation at the cloud based controller employs a hardware-based Trusted Execution Environment (TEE) such as Intel SGX, as demonstrated in \cite{naseri2021securing}. Nevertheless, the latest attacks on SGX have revealed that hardware remains vulnerable to compromise and this could lead to jeopardize both the TEE and the overall systems' security.
Although advanced cryptographic techniques for verifiable computation(e.g., see 
\cite{costello2015geppetto,parno2016pinocchio,bunz2018bulletproofs}, can be used to verify the integrity of outsourced computations, integrating them with the complex structure of HE schemes remains a challenge \cite{fiore2014efficiently,fiore2020boosting}. The (Boolean) circuit representations for control theory algorithms are too complex to implement using these techniques. Additionally, certain techniques like interactive proofs demand multiple rounds of communication between prover and verifier.
In the context of CPS, Cheon et al. \cite{cheon2020authenticated}  proposed a verifiable computation scheme which enables the actuator to verify the non-encrypted controller’s computation. However their scheme is only applicable to linear systems and it cannot be used for encrypted control scenarios. Mahfouzi et al.   \cite{mahfouzi2021secure}  presented a proof of concept of an industrial controller in a cloud-based verifiable computation framework using the Pequin, an end-to-end toolchain for verifiable computation, SNARKs, and probabilistic proofs. However, the authors did not consider preserving the confidentiality of messages transmitted between the local controller and the cloud.

\subsection{Paper's contribution}
The state-of-the-art lacks solutions capable of detecting  attacks against the integrity of cloud-based encrypted controllers. In this paper, we present a verifiable computing solution inspired by the cut-and-choose cryptographic technique, commonly employed in secure multiparty computation protocols \cite{Crépeau2005}, and cut-and-choose auditing \cite{chaum2008scantegrity}. In this approach, the party wishing to perform a secure computation first generates a set of computations and then sends them to a verifier, who proceeds to randomly select a subset from the set and asks the party to disclose the input and output of each computation. In the event that the computations within the chosen subset prove to be correct, the verifier can trust that the remaining computations are correct too, enabling the party to proceed with the secure computation. However, if any of the computations in the subset are found to be incorrect, the protocol is terminated. Our contributions can be summarized as follows.
\begin{itemize}
	\item  We propose a practical verifiable computing solution for encrypted control systems.  The proposed probabilistic scheme allows the plant’s actuator to efficiently validate the computations accomplished by the cloud-based encrypted networked controller without compromising the control system's performance.  Moreover,   we prove that the proposed solution ensures asymptotic detection of any integrity attack. These results are achieved by leveraging two essential ingredients: (i) an offline pre-computed set of control signals complying with the control law for a given set of decoy input parameters, and (ii) a semantically secure encryption mechanism which does not allow the adversary to distinguish the decoys from the actual measurements/control inputs.
	\item  We experimentally validate the performance and effectiveness of the proposed approach using a remotely maneuvered Khepera IV robot\footnote{\url{http://www.k-team.com/khepera-iv}}.  
\end{itemize}

\section{Preliminaries and Problem Formulation}\label{sec:preliminary_knowledge}

Given a vector $v$ and a matrix $M,$ $v_i$ denotes the $i-th$ element of $v$ and $M_{i,j}$ the $(i,j)-$entry of $M.$ Moreover, $I_{2}$ denotes a $2\times 2$ identity matrix while ${\bf{1}}_2$ and ${\bf{0}}_2$ denote column vectors of size two where both elements are equal to $1$ and $0,$ respectively.  Given a number $z$ in the message space of a cryptosystem, we denote the encryption and decryption procedures as $\Enc{z}$ and $\Dec{z}$, respectively.  Consider a  random variable $\mathcal{X}$ with binary possible outcomes ``\textit{success}'' and ``\textit{fail}.'' The probability of getting a 	``\textit{success}'' in a single trial is denoted as $p,$ while the probability of  $k+1$ successes in $k+1$ binary independent experiments is denoted as $p_{[0,\,k]}.$ 

\begin{definition}\label{def:homomorphic_cryptosystem}
	Homomorphic Encryption (HE) refers to a particular class of encryption mechanisms that enables mathematical operations (limited in their type and/or number) to be carried out directly on encrypted data.
	\hfill $\Box$
\end{definition}

\begin{definition}\label{homomorphic_definitions}
	Consider two arbitrary numbers, $z_1$ and $z_2$, in the message space of the cryptosystem. A cryptosystem is said to be \textit{multiplicatively homomorphic} if there exists an operation ``$\otimes$'' allowing encrypted multiplications, i.e., $ z_1z_2=\Dec{\Enc{z_1}\otimes \Enc{z_2}}.$ Analogously, a cryptosystem is said to be \textit{additively homomorphic} if there exists an operation ``$\oplus$''  such that encrypted additions can be performed, i.e., $z_1+z_2=\Dec{\Enc{z_1}\oplus\Enc{z_2}}.$ 
\end{definition}

\subsection{Networked control system setup}\label{sec:proposed_solution}

\begin{figure}[h!]
	\centering
	\includegraphics[width=0.75\columnwidth]{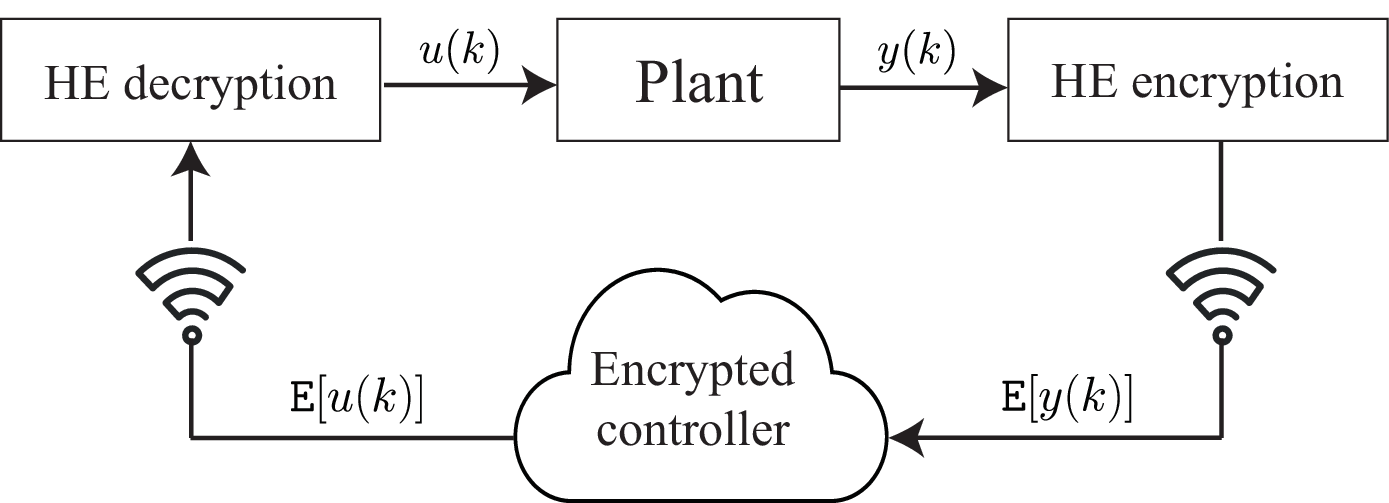}
	\caption{Cloud-based encrypted control scheme}
	\label{fig:encrypted_control_scheme}
\end{figure}

Of interest is the class of networked control systems shown in Fig.~\ref{fig:encrypted_control_scheme}, where the plant's dynamic is described by the discrete-time model
\begin{equation}\label{eq:plant_general}
	\begin{array}{c}
		x(k+1)=f(x(k),u(k)) ,\quad 
		y(k)= g(x(k),u(k)) 
	\end{array}
\end{equation}
with  $k\in \Z_{+}=\{0,1,\ldots\}$  the discrete-time index,    $x\in \rr^n,u\in \rr^m, y\in \rr^p$ the state, input, and output vectors, and $f:  \rr^n \times \rr^m \rightarrow \rr^n$, $g:  \rr^n \times \rr^m \rightarrow \rr^p.$ Moreover, a semantically secure \cite{smart2016cryptography}  randomized HE scheme is used to transmit signals between the plant and the cloud-based controller, allowing the control logic to be executed directly on the encrypted data. The considered dynamic output feedback tracking controller presents the following plaintext dynamic
\begin{equation} \label{eq:controller_general}
	\begin{array}{rcl}
		x_c(k+1)&=&f_c(x_c(k),y(k),r(k))  \\
		u(k) &=& g_c(x_c(k),y(k),r(k)) 
	\end{array}
\end{equation}
where $x_c\in \rr^{n_c}$ is the state of the controller, $r(k) \in \rr^p$ the reference signal, and
$f_c:  \rr^{n_c} \times \rr^p   \times \rr^p \rightarrow \rr^{n_c}$, $g_c:  \rr^{n_c} \times \rr^p  \times \rr^p \rightarrow \rr^m$  the control strategy.

\begin{assumption}
	The control  \eqref{eq:controller_general}  can be implemented on the encrypted data. 
	\hfill $\Box$
\end{assumption}

In what follows, the encrypted-version of \eqref{eq:controller_general} is generically described by:

\begin{equation} \label{eq:controller_general_encrypted}
	\begin{array}{rcl}
		\Enc{x_c(k+1)}&=&f_c^E(\Enc{x_c(k)},\Enc{y(k)},\Enc{r(k)})  \\
		\Enc{u(k)} &=& g_c^E(\Enc{x_c(k)},\Enc{y(k)},\Enc{r(k)}) 
	\end{array}
\end{equation}
where $f_c^E(\cdot,\cdot,\cdot)$ and $g_c^E(\cdot,\cdot,\cdot)$ denote the encrypted controller's operations corresponding to $f_c(\cdot,\cdot,\cdot)$ and $g_c(\cdot,\cdot,\cdot)$.

\subsection{Threat model and Problem formulation}\label{sect:threatModel}

An adversary wants to replace the executed encrypted control operations \eqref{eq:controller_general_encrypted}, i.e.,  
\begin{equation}\label{eq:control_adversary}
	\left\{f_c^E(\cdot,\cdot,\cdot),g_c^E(\cdot,\cdot,\cdot)\right\} \rightarrow  \left\{f_a^E(\cdot,\cdot,\cdot),g_a^E(\cdot,\cdot,\cdot)\right\}
\end{equation}
where
$\left\{f_a^E(\cdot,\cdot,\cdot),g_a^E(\cdot,\cdot,\cdot)\right\}\neq\left\{f_c^E(\cdot,\cdot,\cdot),g_c^E(\cdot,\cdot,\cdot)\right\}$ define the adversary encrypted control logic. Consequently, the encrypted $\Enc{u(k)}$ and $\Enc{x_c(k+1)}$ signals received by the plant from the cloud might be the result of the adversary control logic \eqref{eq:control_adversary}. 
In this work, we assume that the attacker's actions do not force the plant to reach a configuration outside of working domain of the legitimate control law \eqref{eq:controller_general_encrypted}. Consequently, the considered class of attacks is sufficiently intelligent to be stealthy against safety-preserving mechanisms on the plant's side, see e.g., \cite{escudero2023safety}.  
\begin{problem}\label{problem_def}
	\it 
	Given the encrypted networked control architecture shown in Fig.~\ref{fig:encrypted_control_scheme}, design a secure verifiable computing solution capable of assessing the integrity of the encrypted control logic \eqref{eq:controller_general_encrypted}. Moreover, the solution must perform computations in real-time, avoid interference with control actions, and probabilistically ensure the absence of stealthy control logic alterations over time.
\end{problem}

\section{Proposed solution}\label{SecProposedSolution}

\begin{figure}[h!]
	\centering
	\includegraphics[width=0.95\columnwidth]{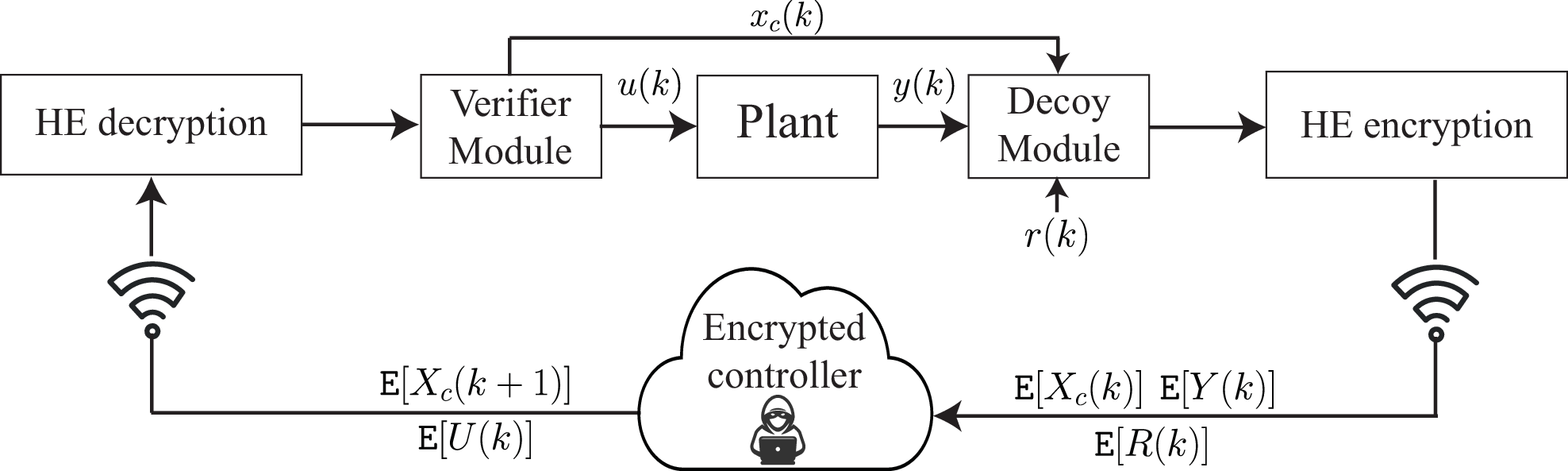}
	\caption{Proposed verifiable computing scheme for encrypted control systems.}
	\label{fig:verifiable_computing_scheme}
\end{figure}

The proposed solution (see Fig.~\ref{fig:verifiable_computing_scheme}) is inspired by the cut-and-choose cryptographic technique, commonly employed in secure multiparty computation protocols \cite{Crépeau2005}. 
For the considered encrypted networked control system setup, we customize the cut-and-choose paradigm as follows:
\begin{itemize}
	\item The cloud-based encrypted controller is the party performing the agreed upon  encrypted control system logic, as described by equation \eqref{eq:controller_general_encrypted}.
	\item A decoy module, local to the plant, is in charge of generating a set of  $n_d>0$  decoy (e.g., artificial) measurement vectors, namely $y_d^i(k),\,i=1,\ldots,n_d.$ The $n_d+1$  measurement vectors (decoys + actual measurement) are encrypted and sent to the controller. The controller executes \eqref{eq:controller_general_encrypted} for each received measurement vector and produces a set of $n_d+1$ computations (i.e., encrypted control inputs).
	\item The verifier is a subsystem local to the plant checking the validity of the decrypted control input vectors obtained from the cloud for all $n_d$ decoy measurement vectors (disclosed to the verifier by the decoy subsystem). If all the decoy computations are deemed accurate, then $u(k)$ is assumed correct and applied to the plant.
\end{itemize}

Moreover, to ensure that the proposed decoy-based verifiable computing paradigm provides a solution to Problem~\ref{problem_def}, the following objectives must be achieved:

\begin{itemize}
	\item [(]$\!\!\!O_1)\!$ The verification of the computations for each decoy must be real-time affordable and not require the online local re-computation of the control algorithm itself.
	\item [(]$\!\!\!O_2)\!$ The decoy measurements should not interfere with the computation of the control action $\Enc{u(k)}$ for the actual sensor measurements $\Enc{y(k)}$.     
	\item [(]$\!\!\!O_3)\!$ The attacker's probability $p_{[0,\,k]}$ of changing the control logic as in \eqref{eq:control_adversary} without being detected asymptotically converges to zero, i.e. $\lim_{k\rightarrow \infty}p_{[0,\,k]}=0.$
	
\end{itemize}

\subsection{Proposed verifiable computing scheme}\label{sec:proposed_ver_comp}
In what follows, first the proposed verifiable computing scheme is presented and then its effectiveness in addressing the points $(O_1)$-$(O_3)$ is proven.

$\circ$  For each required computation, the encrypted controller receives from the plant the set $\{\Enc{{x}_c(k)},\Enc{y(k)},\Enc{r(k)}\}.$ The encrypted controller first resets its state to $\Enc{{x}_c(k)}$, then computes the updated internal state $\Enc{x_c(k+1)}$ and control input $\Enc{u(k)}$ which are both transmitted to the plant. If the computation is pertaining to the real measurement, the plaintext corresponding to $\Enc{{x}_c(k)}$   is equal to the one provided by the controller at the previous iteration.

$\circ$ Offline, a set $\mathcal{D}$ of $N_d$ tuples $(\bar{u}_d,\bar{x}_{c_d}^+,\bar{x}_{c_d},\bar{y}_d,\bar{r}_d),$ complying with the controller's logic \eqref{eq:controller_general} is defined, i.e.,
\begin{equation}\label{eq:pre-computed-decoy-set}
	\begin{array}{rcl}
		\mathcal{D}\!\!\!\!&=&\!\!\!\!\!\{ (\bar{u}_d^j,\bar{x}_{c_d}^{j^+},\underbrace{\{\bar{x}_{c_d}^j,\bar{y}^j_d,\bar{r}^j_d\}}_{\text{decoy}}), j=1,\ldots,N_d :\vspace{-0.0cm}\\
		&&    \underbrace{\bar{u}^j_d=g_c(\bar{x}_{c_d}^j,\bar{y}^j_d,\bar{r}^j_d)}_{\text{expected control input}} \},\,
		\underbrace{\bar{x}_{c_d}^{j^+}=f_c(\bar{x}_{c_d}^j,\bar{y}^j_d,\bar{r}^j_d)}_{\text{expected next state}} \}
	\end{array}
\end{equation}

$\circ$ At each $k\geq 0$, a set of $n_d\geq 1$ decoys is selected (with possible repetitions) from $\mathcal{D}$ and freshly encrypted into $\{\Enc{{x}_{c_d}^i(k)},\Enc{y_d^i(k)},\Enc{r_d^i(k)}\}_{i=1}^{n_d}$. 

$\circ$ At each $k\geq 0$ the $n_d$ decoys are randomly arranged with the actual vector $\{x_c(k),y(k),r(k)\}$ to form the  following encrypted matrices which are sent to the controller (see Fig.~\ref{fig:verifiable_computing_scheme}):
\begin{equation}\label{eq:transmitted_input_decoy}
	\!\!\begin{array}{rcl}
		\Enc{X_c(k)}&\!\!\!\!:=&\!\!\!\!
		\left[
		\Enc{x_c(k)},\Enc{x_{c_d}^1(k)},\ldots,\Enc{x_{c_d}^{n_d}(k)}
		\right] 
		\Omega_k \\
		\Enc{Y(k)}&\!\!\!\!:=&\!\!\!\!
		\left[
		\Enc{y(k)},\Enc{y_d^1(k)},\ldots,\Enc{y_{d}^{n_d}(k)}
		\right] 
		\Omega_k \\
		\Enc{R(k)}&\!\!\!\!:=&\!\!\!\!
		\left[
		\Enc{r(k)},\Enc{r_d^1(k)},\ldots,\Enc{r_{d}^{n_d}(k)}
		\right] 
		\Omega_k \\
	\end{array} 
\end{equation}
with $\Omega(k)\in \rr^{(n_d+1) \times (n_d+1)}$ a random permutation matrix. On the other hand, the controller, without the knowledge of $\Omega_k,$ performs the encrypted control algorithm \eqref{eq:controller_general_encrypted} for each column of the received matrices and arranges the outputs into the following, which are sent to the plant (see Fig.~\ref{fig:verifiable_computing_scheme}):
\begin{equation}\label{eq:transmitted_output_decoy}
	\!\!\!\!\!
	\begin{array}{rcl}
		\Enc{U(k)}&\!\!\!\!:=&\!\!\!\!
		[
		\Enc{u(k)},\Enc{u_{d}^1(k)},\ldots,\Enc{u_{d}^{n_d}(k)}
		] 
		\Omega_k \\
		\Enc{X_c(k+1)}&\!\!\!\! :=&\!\!\!\!
		[
		\Enc{x_{c}(k+1)},\Enc{x_{c_d}^{1}(k+1)},\\&&\!\!\!\!
		\qquad \qquad \qquad\ldots,\Enc{x_{c_d}^{n_d}(k+1)}
		] 
		\Omega_k 
	\end{array} 
\end{equation}

$\circ$  The verifier decrypts \eqref{eq:transmitted_output_decoy} and  checks, for each used decoy  $\{{x}_{c_d}^i(k),{y_d^i(k)},{r_d^{i}(k)}\}_{i=1}^{n_d},$ if the computed control inputs and next states are equal to the expected values $\{u^i_d(k)\}_{i=1}^{n_d}$  and  $\{x^i_{c_d}(k+1)\}_{i=1}^{n_d}$  in \eqref{eq:pre-computed-decoy-set}. 

\begin{proposition}\label{prop:absence_of_stealthy}
	\it 
	If $N_d \geq 2,$ the proposed verifiable computing scheme achieves the objectives $(O_1)$-$(O_3).$
\end{proposition}

\begin{proof}The proof can be divided into three parts:
	
	\noindent - $(O_1):$ by constructions, the decoy set \eqref{eq:pre-computed-decoy-set} is pre-computed. Therefore, for each used decoy $\{\Enc{{x}_{c_d}^i(k)},\Enc{y_d^i(k)},\Enc{r_d^{i}(k)}\},$ the verification of $u_{d}^i(k)$ and $x_{c_d}^{i}(k+1)$ prescribes only their comparison with the pre-computed value  in \eqref{eq:pre-computed-decoy-set}. Consequently, the computational overhead of the proposed solution is mainly related to the encryption and decryption of the decoys.
	
	\noindent - $(O_2):$ The controller resets the internal state before each computation. Moreover, for the actual measurement, the reset value is equal to the controller's state at the previous iteration. Consequently, regardless of $n_d,\, N_d,\,\Omega_k,$ the computations pertaining to  the decoys do not affect $\Enc{u(k)}.$ 
	
	\noindent - $(O_3):$  
	If the attacker is only interested in arbitrarily affecting the control logic and $n_d=1, N_d=1$ (i.e., $1$ decoy from a pool of a single decoy), then the attacker can guess at $k=0$ the decoy measurement vector and use the valid decoy output $\Enc{u_d^1(0)}$ to send to the plant the matrix $\Enc{U(k)}=[\Enc{u_d^1(0)}, \Enc{u_d^1(0)}]\Omega_k,\,\forall\,k$ which would trigger an anomaly with probability $0.5,$ irrespective of $k$. To avoid copying the same ciphertext, which could lead to trivial detection, the attacker can re-encrypt it. On the other hand, for $N_d= 2,$ or any other greater value,  the probability that an attacker can, in a single trial, successfully change  $\Enc{u(k)}$ while leaving the decoy computation unchanged is equal to the probability of guessing in which column of $\Enc{Y(k)}$ the measurement vector $\Enc{y(k)}$ is contained. Given the randomized nature of the used homomorphic encrypting scheme and $\Omega_k,$ each column vector has an equal probability, resulting in a Bernoulli probability of cheating
	$p=1/(n_d+1).$ Consequently,  in the  interval $[0, k]$,  the probability $p_{[0,\,k]}$ that the attacker remains always undetected follows the Binomial distribution $(\frac{1}{n_d+1})^{k+1}$ that asymptotically converges to zero as $k\rightarrow \infty.$ 
\end{proof}

\section{Proof-of-concept validation using a remotely controlled mobile robot}\label{sec:proof}

In what follows, by considering as case of study a remotely controlled differential-drive robot, we verify the effectiveness of the proposed verifiable computing scheme. 

\subsection{Robot Model}
\label{subsec. differential_drive_robot_model}

We consider a differential-drive robot equipped with two real independently driven wheels and a front castor wheel for body support.  The pose of the robot is described by the planar coordinates $(p_x,p_y)$ of its center of mass and orientation $\theta$.
By resorting to the forward Euler discretization method and a sampling time $t_s>0,$  the discrete-time kinematic model of the differential-drive is:
\begin{equation}\label{eq:robot_discrete_time_equation}
	\begin{array}{rl}
		p_x(k+1) = & p_x(k) \!+\! \frac{t_s R}{2}\cos{\theta(k)}(\omega_r(k) + \omega_l(k)) \vspace{0.1cm}\\
		p_y(k+1) = & p_y(k) \!+\! \frac{t_s R}{2}\sin{\theta(k)}(\omega_r(k) + \omega_l(k)) \vspace{0.1cm}\\ 
		\theta(k+1) = & \theta(k) \!+\! \frac{t_s R}{D}(\omega_r(k) - \omega_l(k)) \\
	\end{array}
\end{equation}
where $R>0$ is the radius of the wheels, $D>0$ the rear axle length, and $u^{D}=[\omega_r,$ $\omega_l]^T\in \rr^2$ the control input vector, which consists of the angular velocities of the right and left wheel, respectively. 
Moreover, $x(k)=[p_x(k), p_y(k), \theta(k)]^T\in \rr^3$ denotes the robot's state vector. 

By denoting with $v(k)$ and $\omega(k)$ the linear and angular velocities of the center of mass of the robot, 
it is possible to apply to \eqref{eq:robot_discrete_time_equation} the static transformation
\begin{equation}\label{eq:from_diff-drive_to_unicycle}
	\left[\begin{array}{c}
		v(k)  \\
		\omega(k) 
	\end{array}\right] = H \left[\begin{array}{c} \omega_r(k)\\
		\omega_l(k)  \end{array}\right],\quad
	H:=\left[\begin{array}{cc}
		\frac{R}{2} &\frac{R}{2}\vspace{0.1cm}\\
		\frac{R}{D} &\frac{-R}{D}
	\end{array}\right], 
\end{equation}
and describe the robot's behavior by means of the following discrete-time unicycle model:
\begin{equation}\label{eq:robot_unicycle-model}
	\begin{array}{rl}
		{p}_x(k+1)=&{p}_x(k)+t_s v(k)\cos{\theta(k)} \vspace{0.1cm} \\
		{p}_y(k+1)=&{p}_y(k)+t_s v(k)\sin{\theta(k)} \vspace{0.1cm} \\
		{\theta}(k+1)=&{\theta}(k)+T\omega(k)
	\end{array}
\end{equation}
Since \eqref{eq:robot_unicycle-model} has been obtained using an Euler forward discretization and a static transformation, we can exploit the commutative property between feedback linearization and the input-output linearization used in  \cite[Property 1]{tiriolo2022receding} to linearize the discrete-time model of the unicycle. In particular, by considering a small scalar $b>0$ and two virtual outputs
\begin{equation}\label{eq:input-outout-transf}
	y(k)=\left[\begin{array}{cc}
		p_x(k)+b\cos\theta(k),\,
		p_y(k)+b\sin\theta(k) 
	\end{array}\right]^T, 
\end{equation}
representing the coordinates of a fictitious point $\bf B$ displaced at a distance $b$ from $[p_x,p_y]^T,$ the  state-feedback law
\begin{equation}\label{feedback-linearization}
	\left[\!\!
	\begin{array}{c}
		v\\
		\omega
	\end{array}
	\!\!\right]
	\!=\!
	T_{FL}(\theta)\!
	\left[\!
	\begin{array}{c}
		u_1\\
		u_2
	\end{array}
	\!\right],\,\,	
	T_{FL}(\theta)\!=\!\left[\!\!\!
	\begin{array}{cc}
		\cos\theta&\sin\theta\\
		\frac{-\sin\theta}{b}&\frac{\cos\theta}{b} \\
	\end{array}\!\!
	\right]
\end{equation}
recasts  \eqref{eq:robot_unicycle-model} into a two-single integrator model with a decoupled nonlinear internal dynamics  \cite[Property 1]{tiriolo2022receding}:
\begin{subequations}\label{eq:input-output-lin-sys}
	\begin{gather}
		y(k+1)=Ay(k)+Bu(k),\,\, A=I_{2},\,\, B=TI_{2} \label{eq:linear_model}\\
		\theta(k+1)=\theta(k)+T\frac{-\sin\theta(k)u_1(k)+\cos\theta(k)u_2(k)}{b} \label{eq:internal-dynamics}
	\end{gather}
\end{subequations}
where $u(k)=[u_1(k),\,u_2(k)]^T\in \rr^2$  are the control inputs of the feedback-linearized robot model.  As stated in  \cite[Remark~1]{tiriolo2022receding}, any  linear tracking controller for \eqref{eq:input-outout-transf} allows $y(k)$ to track any reference trajectory with a stable internal dynamics for $\theta(k).$

\subsection{HE cryptosystem}
\label{subsubsec. Paillier}
The data exchanged with the controller are encrypted using the Paillier cryptosystem \cite{paillier1999public}. The  Paillier cryptosystem is a randomized and semantically secure partially homomorphic encryption scheme which allows encrypted additions of two ciphertexts (see Definition~\ref{homomorphic_definitions}) and multiplications of a ciphertext by a plaintext number, i.e., $z_1z_2=\Dec{\Enc{z_1}\otimes z_2}.$

\subsection{Tracking controller}
\label{subsec. PI_controller}
Consider a desired robot's reference pose $p^r(k)=[p_x^r(k),p_y^r(k), \theta^r(k)]^T \in \rr^3,\forall\,k\geq 0$ and the associated 2D reference trajectory, namely $r^{B}(k),$ for the $\bf B$ point
\begin{equation}\label{eq:ref_timing_law}
	r(k)=[p_x^r(k) +b\cos \theta^r(k),\,p_y^r(k)+b \sin \theta^r(k)]^T,\,\, \forall \,k\geq 0.
\end{equation}
A standard discrete-time Proportional Integral (PI) controller is used to allow the $\bf B$  point to track $r(k).$  The PI control's law can be written as in \eqref{eq:controller_general}, resulting in the following state-space representation
\begin{equation}\label{eq:PI_state_space}
	\begin{array}{rcl}
		x_{c}(k+1)&=&x_c(k)+T_s (r(k)-y(k))\\
		u(k)&=& K_i x_c(k) + K_p(r(k)-y(k)) 
	\end{array}
\end{equation}
where $x_c\in \rr^2,\,K_p=k_pI_2\in \rr^{2\times 2}$ is the controller integral state vector, $T_s=t_sI_2\in \rr^{2\times 2},$ and $K_i=k_iI_2\in \rr^{2\times 2},k_p,k_i\in \rr$ are the controller's gains. Equivalently, \eqref{eq:PI_state_space} can be represented as the following input-output linear relation 
\begin{equation}\label{eq:PI_matrix_form}
	\begin{bmatrix}
		x_c(k+1) \\ u(k) 
	\end{bmatrix} = \underbrace{\begin{bmatrix}
			I_2 & T_s & -T_s  \\
			K_i & K_p & -K_p
	\end{bmatrix}}_{\text{$K$}}
	\begin{bmatrix}
		x_c(k) \\ r(k) \\ y(k)
	\end{bmatrix}.
\end{equation}
To allow the control law \eqref{eq:PI_matrix_form} to be executed in encrypted form on the cloud, the matrix  $K \in \rr^{4\times 6}$ is there pre-uploaded and available in plaintext (non-encrypted matrix). On the other hand, the vector $v(k)=\begin{bmatrix}
	x_c(k) & r(k) & y(k)\end{bmatrix}^T\in \rr^6,$ before transmission to the cloud, is encoded. Consequently,  the control law
\eqref{eq:PI_matrix_form} is computed on the cloud, in an encrypted form, as: 
\begin{equation}\label{eq:encryptedPI}
	\!\!\!\!
	\begin{array}{rcr}
		\Enc{x_{c,l}(k+1)}\!\!&\!\!\!\!\!=\!\!\!\!\!\!\!&\!\!(V_{l,1} \otimes \Enc{v_1(k)})\!\oplus\!\!\ldots\!\!\oplus\! (V_{l,6} \otimes
		\Enc{v_6(k)})\\
		\Enc{u_l(k)}\!\!&\!\!\!\!\!=\!\!\!\!\!\!\!&\!\!(Z_{l,1} \otimes \Enc{v_1(k)})\!\oplus\!\!\ldots \!\!\oplus \! (Z_{l,6} \otimes
		\Enc{v_6(k)}),\\
		&& l=1,2,
	\end{array}
\end{equation}
where 
$V=\begin{bmatrix}
	I_2 & T_s & -T_s\end{bmatrix}\in \rr^{2\times 6},$  $Z=\begin{bmatrix}
	K_i & K_p & -K_p\end{bmatrix}\in \rr^{2\times 6}.$ 
On the plant's side, after decryption, i.e., $u(k)= \left[\Dec{\Enc{u_1(k)}},\Dec{\Enc{u_2(k)}}\right]^T,$ the  angular velocities of right and left wheels are recovered as
\begin{equation}\label{eq:transfomation_of_inputs}
	u^D(k)=H^{-1}T_{FL}(\theta(k))u(k).
\end{equation}

\subsection{Setup and experimental results}\label{sec:setupandexperiments}
\begin{figure}[h!]
	\centering
	\includegraphics[width=0.8\columnwidth]{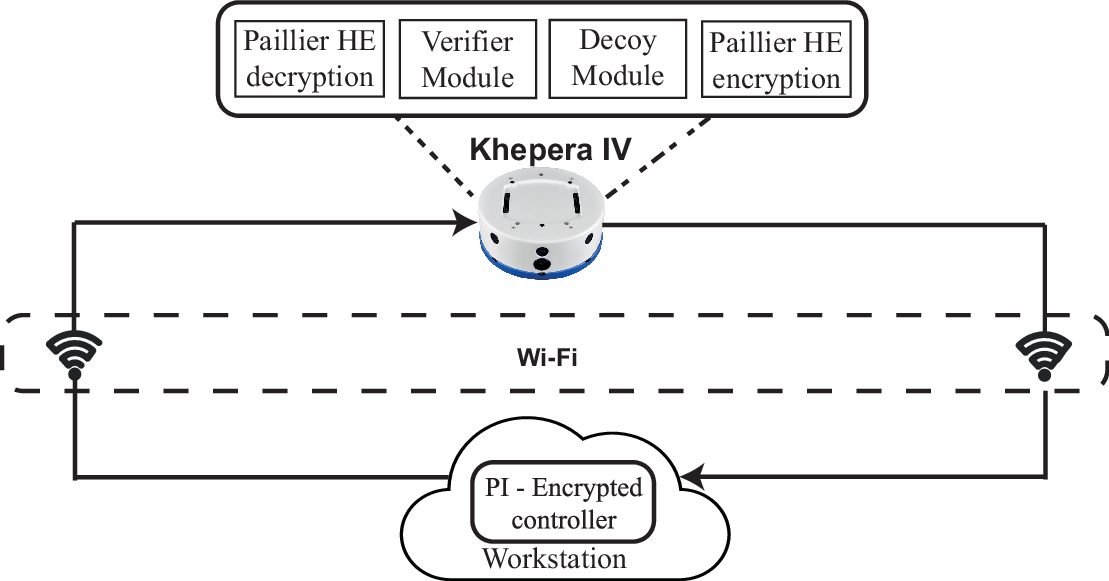}
	\caption{Experimental setup.}
	\label{fig:Experimental_utilized_setup}
\end{figure}
The used experimental networked control system setup is illustrated in Fig.~\ref{fig:Experimental_utilized_setup}, where a Khepera IV robot is remotely and wirelessly controlled by a workstation acting as the remote cloud.  The Khepera IV is a differential-drive robot with $R=0.0210[m]$ and $D=0.10470[m],$ and it is equipped with an 800MHz ARM Cortex-A8 Processor with C64x Fixed Point DSP core and 256 MB of RAM. The cloud is simulated using a workstation running Windows 10 with an Intel Core i9-13900KF processor with 24 cores, where the robot's encrypted control law \eqref{eq:encryptedPI} is executed. The Verifier, Decoder, Encryption and Decryption modules and \eqref{eq:transfomation_of_inputs} are all executed on the Khepera robot's. The remote exchange of encrypted data between the robot and the workstation has been implemented using the 802.11 b/g WiFi protocol has been used.
The encrypted control law has been implemented using Phyton and the \textit{eclib} library (\url{https://pypi.org/project/eclib/}) and using as control knobs $t_s=0.15[\sec],$ $b=0.1,$  $k_p=4$, and  $k_i=0.2.$ Moreover,  the  Paillier encryption is performed using $|p|=|q|=512$ and a quantization parameter $\delta=0.0001.$
To speed-up the arithmetic operations necessary to perform encryption, decryption and homomorphic multiplications, the C-coded Python extension module \textit{gmpy2} (\url{https://pypi.org/project/gmpy2/}) has been used.
The reference trajectory \eqref{eq:ref_timing_law}, whose path in the dashed red line shown in Fig.~\ref{fig:crypt_trajectory_with_attack}, has been obtained interpolating a set of $22$ waypoints distributed along the path using a cubic spline configured such that the average longitudinal velocity between any two consecutive waypoints is fixed and equal to $0.09[m/s]$.

The decoy module is configured with $n_d=1,$ $N_d=2,$ i.e., to randomly choose a single decoy from a pool of 2 pre-computed decoy tuples: 
\begin{equation}\label{eq:used-decoy-set}
	\begin{array}{rcl}
		\mathcal{D}\!&=&\! \{ 
		(2 \cdot{} {\bf{1}}_2,0.075\cdot{} {\bf{1}}_2,\underbrace{\{{\bf{0}}_2,{2 \cdot{} \bf{1}}_2,2.5 \cdot{} {\bf{1}}_2\}}_{\text{decoy 1}}),\\
		&&(-3 \cdot{} {\bf{1}}_2,4.85\cdot{} {\bf{1}}_2,\underbrace{\{5\cdot{}{\bf{1}}_2,{\bf{1}}_2,{\bf{0}}_2\}}_{\text{decoy 2}})
		\}.
	\end{array} 
\end{equation}
Consequently, for the given choices, the probability that any attack remains undetected over the discrete time interval $[0, k]$ is $(\frac{1}{2})^{k+1}.$ At each $k,$ the cloud receives
\begin{equation*}
	\begin{array}{c}
		\Enc{X_c(k)}=
		\left[
		\Enc{x_c(k)},\Enc{x_{c_d}^1(k)}
		\right] 
		\Omega_k \\
		\Enc{Y(k)}=
		\left[
		\Enc{y(k)},\Enc{y_d^1(k)}
		\right] 
		\Omega_k \\
		\Enc{R(k)}=
		\left[
		\Enc{r(k)},\Enc{r_d^1(k)}
		\right] 
		\Omega_k
	\end{array} 
\end{equation*}
and without the knowledge of $\Omega_k,$ it executes the encrypted PI control law \eqref{eq:encryptedPI} starting for each column of $\Enc{X_c(k)},\Enc{Y(k)},\Enc{R(k)}.$ Then, all the encrypted outputs are collected into the following vectors 
\begin{equation*}
	\begin{array}{c}
		\Enc{U(k)}=
		[
		\Enc{u(k)},\Enc{u_{d}^1(k)}
		] 
		\Omega_k \\
		\Enc{X_c(k+1)}=
		[
		\Enc{x_{c}(k+1)},\Enc{x_{c_d}^{1}(k+1)}
		] 
		\Omega_k 
	\end{array} 
\end{equation*}
and sent to the robot's Paillier HE decryptor module, which decrypts and sends the plaintext vectors to the Verifier. If at the time instant $\bar{k}\geq 0,$ the Verifier finds an incorrect computation related to the decoy, it will stop the robot in the current position, i.e., $u^D(k)=[0,0]^T,\forall\,k\geq \bar{k}.$

Two experiments have been performed to verify that the proposed verifiable computing protocol \textit{(i)}  does not interfere with tracking operations of the networked control system, \textit{(ii)} and it is able to promptly detect cyber-attacks affecting the integrity of the encrypted control logic. In particular, in the first experiment, we consider a scenario where no cyber-attacks affect the encrypted computation, while in the second, a cyber-attack violates the integrity of the PI Encrypted controller starting from $k=200\, (t=30\,[\sec])$  to disrupt the reference tracking task. We have configured the attacker to disrupt the encrypted PI logic as follows. First, it allows the computation of \eqref{eq:encryptedPI} for the first set of measurements and records the obtained encrypted outputs. Then, instead of re-computing \eqref{eq:encryptedPI} for the second set of data, it forces the output to be equal to the outcome of the first computation. 
As a consequence, $\forall\,k\geq 200$, the transmitted  values are (with an equal probability of 0.5)  \textit{either}  
\begin{itemize}
	\item (\textit{Case 1}): $ 
	\Enc{U(k)}= 
	[\Enc{u_{d}^1(k)},\Enc{u_{d}^1(k)}],$
	$\Enc{X_c(k+1)}=[\Enc{x_{c_d}^{1}(k+1)},\Enc{x_{c_d}^{1}(k+1)}],$ \textit{or}
	\item (\textit{Case 2}): $\Enc{U(k)}=[\Enc{u(k)},\Enc{u(k)}],$
	$\Enc{X_c(k+1)}=[\Enc{x_{c}^j(k+1)},\Enc{x_{c}^j(k+1)}],$ 
\end{itemize}
where only \textit{Case 1} bypasses the verification performed on the decoy. Note that even if the attack is aware of the used decoy set $\mathcal{D}$, the randomized nature of the used cryptosystem does not allow it to distinguish between real and decoy measurements. Consequently, the considered attack scenario is representative of any other deception attack against the integrity of the encrypted controller.

The  experimental results are collected in Figs.~\ref{fig:trajectory}-\ref{fig:CPU_time_boxPlot}. In the subplots \ref{fig:crypt_trajectory_no_attack} and \ref{fig:orientation_without_attack} it is possible to observe that, as expected, the proposed verifiable computing operations do not interfere with the real-time operations of the control loop. Indeed, as shown in the box plot of Fig.~\eqref{fig:CPU_time_boxPlot}, the maximum time to execute the control loop of Fig.~\ref{fig:Experimental_utilized_setup} is equal to $0.0631 [\sec]$ which is  smaller than the used sampling time. Moreover, the computational overhead caused by the proposed scheme, mainly due to the encryption and decryption of the decoys, is  smaller than $0.005\,[\sec]$ with an average of $0.0032\,[\sec].$ On the other hand, the subplots \ref{fig:crypt_trajectory_with_attack} and \ref{fig:orientation_with_attack} show that the Verifier module was able to detect the attack at $t=30[\sec]$ and to stop the robot after such occurrence. The instantaneous detection finds justification in the fact that at $t=30[\sec],$ \textit{decoy 1} was sent along the real measurement vectors while the cloud returned $U(30)=[[0.045,\,0.082]^T], [0.045,\,0.082]^T]$ (\textit{Case 2} for the attacker's strategy) which was not compatible with the expected decoy output.   The demo pertaining to the performed experiments is available at the following YouTube link: 
\href{http://tinyurl.com/4a8u4y9r}{http://tinyurl.com/4a8u4y9r}.

\begin{figure}
	\centering
	\begin{subfigure}[b]{0.48\columnwidth}
		\centering
		\includegraphics[width=0.85\columnwidth]{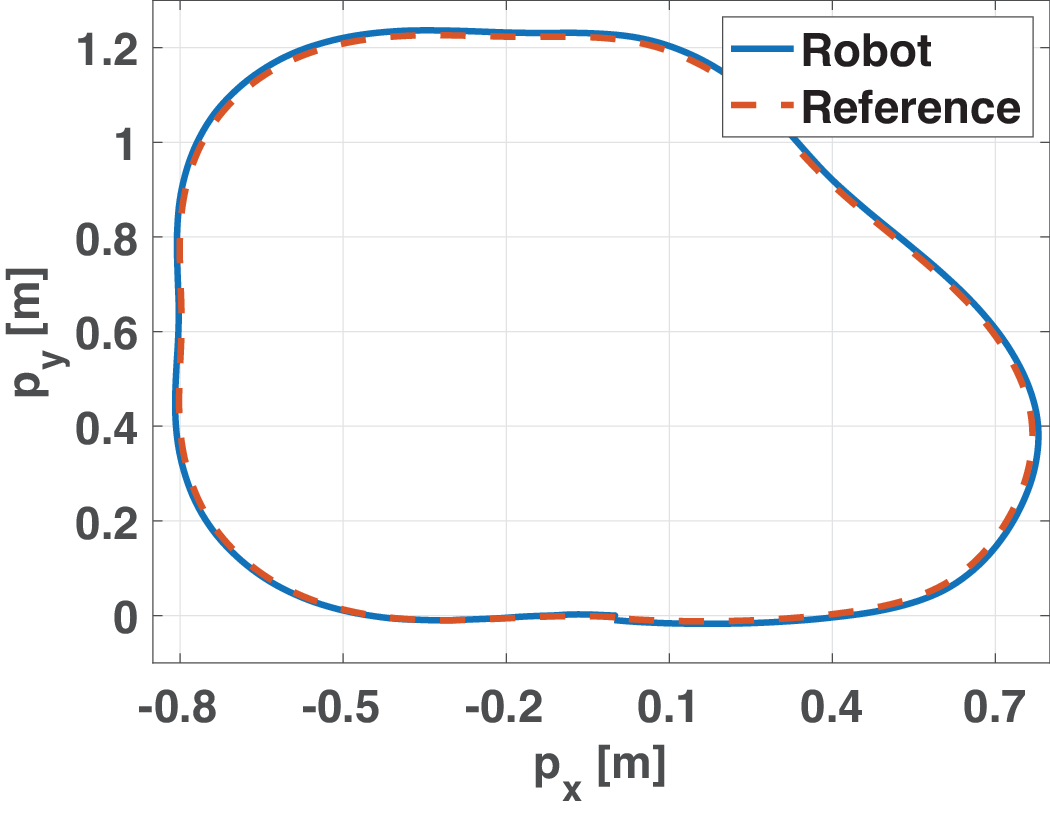}
		\caption{Without attack.}
		\label{fig:crypt_trajectory_no_attack}
	\end{subfigure}
	\hfill
	\begin{subfigure}[b]{0.48\columnwidth}
		\centering
		\includegraphics[width=0.85\columnwidth]{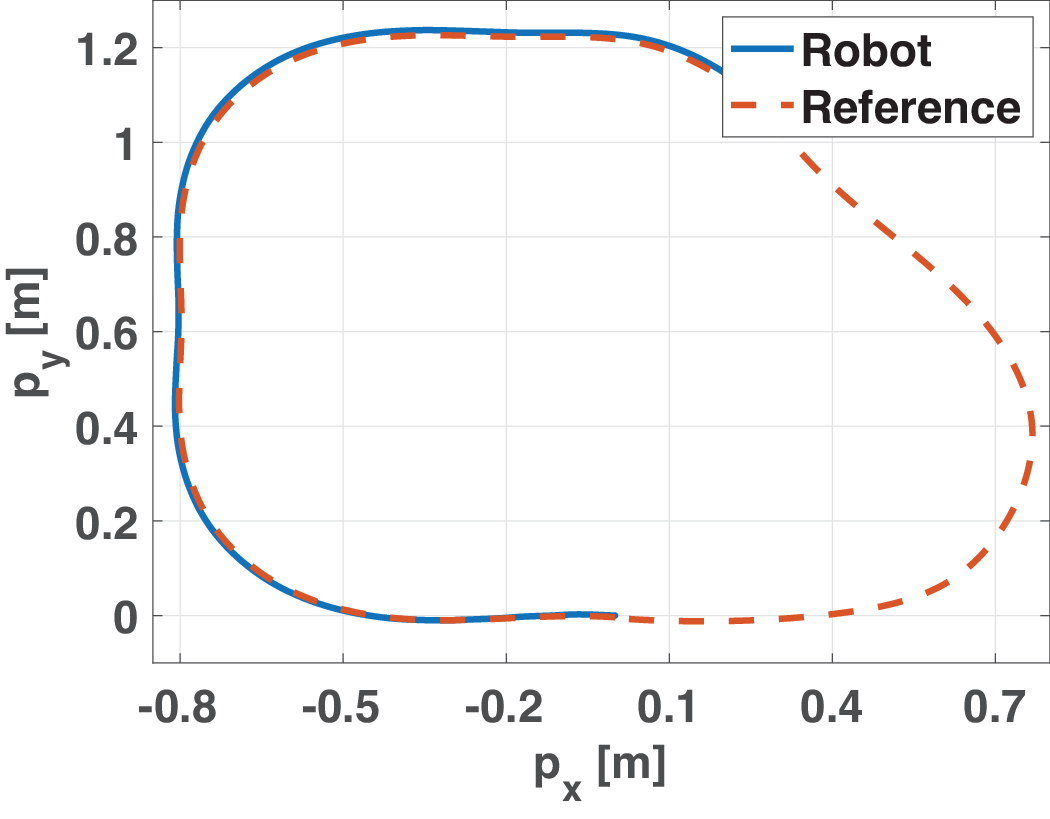}
		\caption{With attacks.}
		\label{fig:crypt_trajectory_with_attack}
	\end{subfigure}
	\hfill
	\caption{Robots trajectory.} 
	\label{fig:trajectory}
\end{figure}

\begin{figure}
	\centering
	\begin{subfigure}[b]{0.48\columnwidth}
		\centering
		\includegraphics[width=0.85\columnwidth]{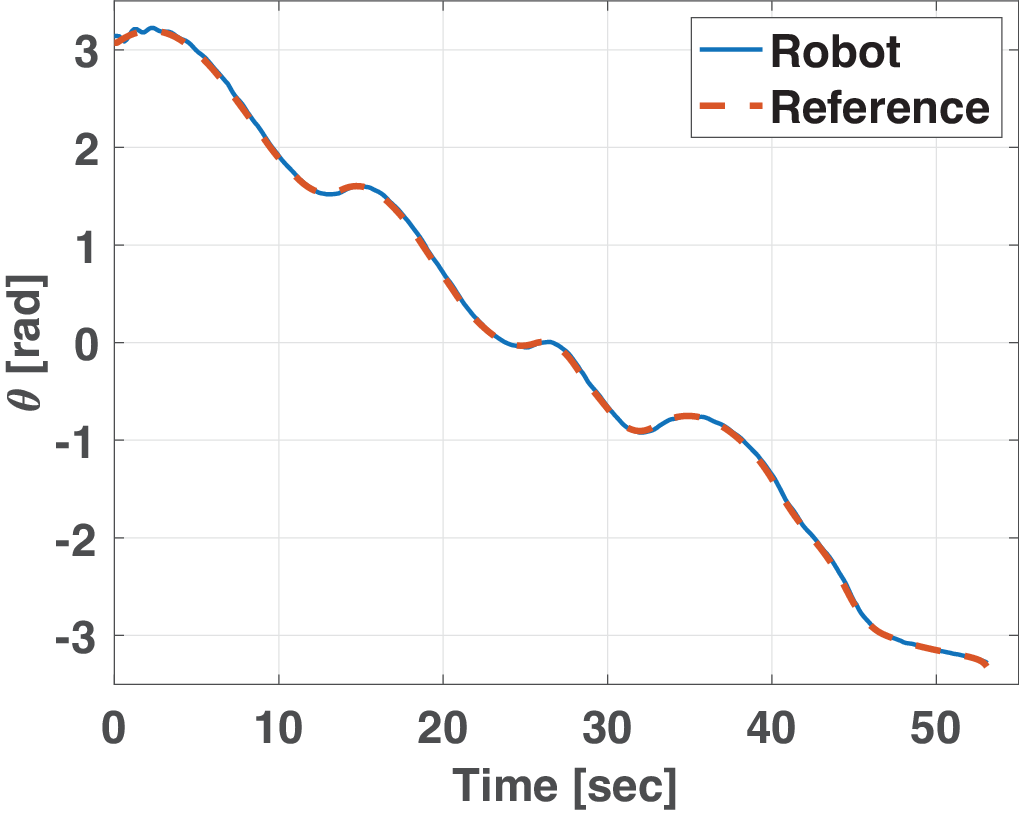}
		\caption{Without attacks.}
		\label{fig:orientation_without_attack}
	\end{subfigure}
	\hfill
	\begin{subfigure}[b]{0.48\columnwidth}
		\centering
		\includegraphics[width=0.85\columnwidth]{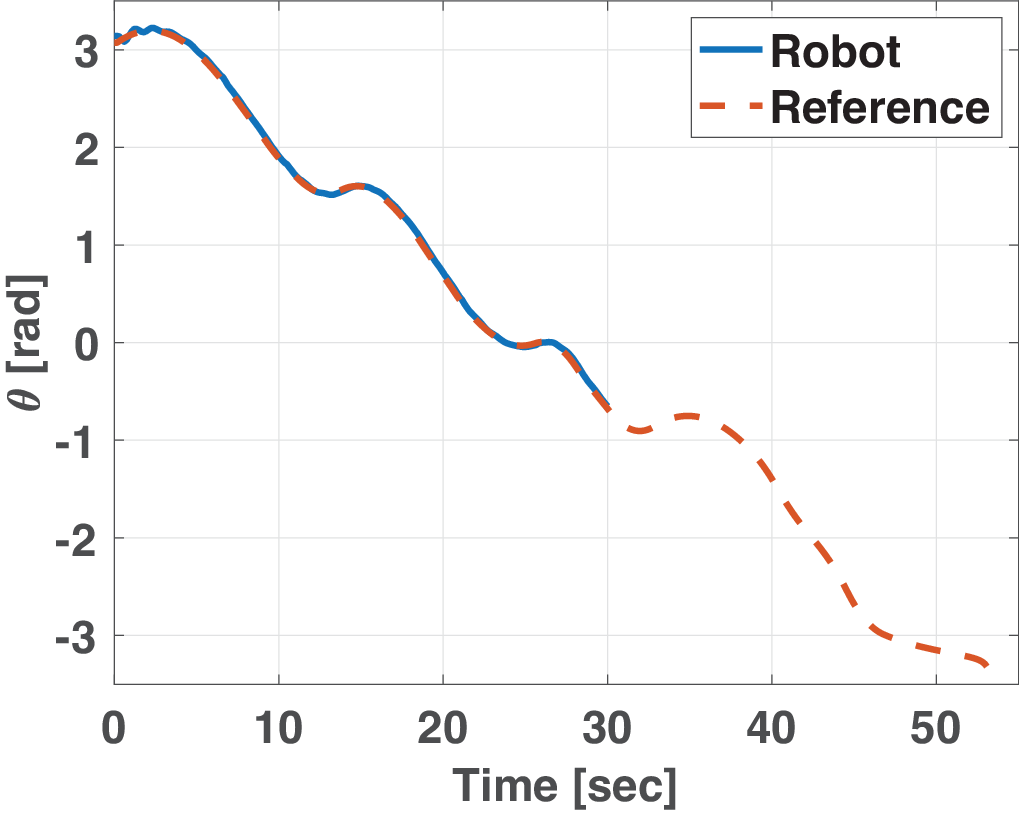}
		\caption{With attacks.}
		\label{fig:orientation_with_attack}
	\end{subfigure}
	\hfill
	\caption{Robot's orientation.} 
	\label{fig:orientation}
\end{figure} 
\begin{figure}[h!]
	\centering
	\includegraphics[width=0.75\columnwidth]{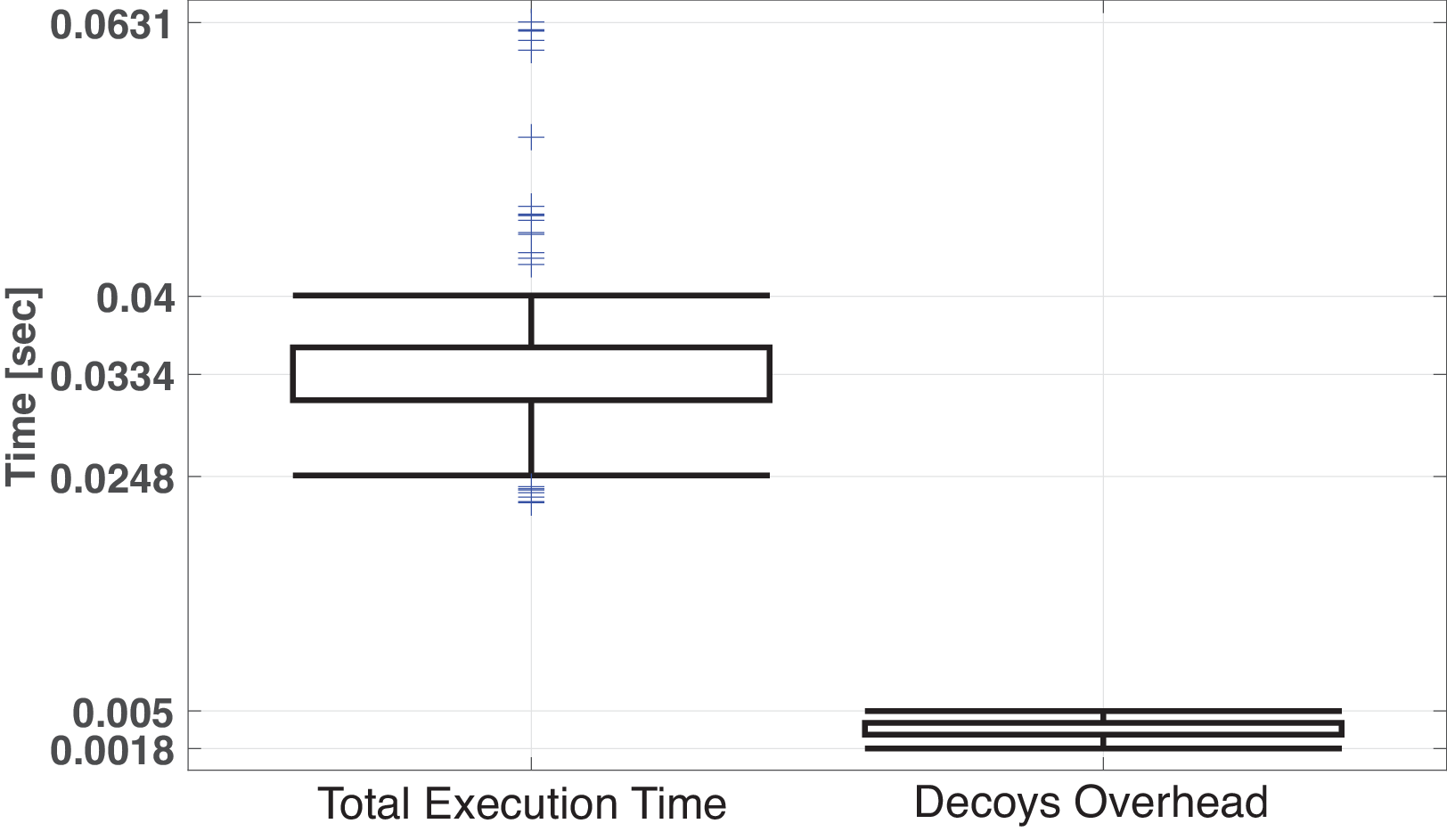}
	\caption{Boxplot for encrypted control loop execution time.}
	\label{fig:CPU_time_boxPlot}
\end{figure}

\section{Conclusions}\label{sec. conclusions}
We presented an efficient verifiable computing solution for encrypted cloud-based control systems. The proposed approach relies on probabilistically checkable proof that enables the plant's actuator to authenticate the computations performed by the encrypted networked controller without compromising the performance of the control scheme. The effectiveness and real-time applicability of the proposed scheme have been demonstrated through experiments using a remotely controlled Khepera-IV differential-drive robot.

\bibliographystyle{IEEEtran}
\bibliography{Ref}

\begin{thebibliography}{10}
\providecommand{\url}[1]{#1}
\csname url@rmstyle\endcsname
\providecommand{\newblock}{\relax}
\providecommand{\bibinfo}[2]{#2}
\providecommand\BIBentrySTDinterwordspacing{\spaceskip=0pt\relax}
\providecommand\BIBentryALTinterwordstretchfactor{4}
\providecommand\BIBentryALTinterwordspacing{\spaceskip=\fontdimen2\font plus
\BIBentryALTinterwordstretchfactor\fontdimen3\font minus
  \fontdimen4\font\relax}
\providecommand\BIBforeignlanguage[2]{{%
\expandafter\ifx\csname l@#1\endcsname\relax
\typeout{** WARNING: IEEEtran.bst: No hyphenation pattern has been}%
\typeout{** loaded for the language `#1'. Using the pattern for}%
\typeout{** the default language instead.}%
\else
\language=\csname l@#1\endcsname
\fi
#2}}

\bibitem{xia2022brief}
Y.~Xia, Y.~Zhang, L.~Dai, Y.~Zhan, and Z.~Guo, ``A brief survey on recent
  advances in cloud control systems,'' \emph{IEEE Trans. on Circuits and
  Systems II: Express Briefs}, vol.~69, no.~7, pp. 3108--3114, 2022.

\bibitem{teixeira2012attack}
A.~Teixeira, D.~P{\'e}rez, H.~Sandberg, and K.~H. Johansson, ``Attack models
  and scenarios for networked control systems,'' in \emph{Int. Conf. on High
  Confidence Networked Systems}, 2012, pp. 55--64.

\bibitem{giraldo2018survey}
J.~Giraldo, D.~Urbina, A.~Cardenas, J.~Valente, M.~Faisal, J.~Ruths, N.~O.
  Tippenhauer, H.~Sandberg, and R.~Candell, ``A survey of physics-based attack
  detection in cyber-physical systems,'' \emph{ACM Computing Surveys (CSUR)},
  vol.~51, no.~4, pp. 1--36, 2018.

\bibitem{kim2016encrypting}
J.~Kim, C.~Lee, H.~Shim, J.~H. Cheon, A.~Kim, M.~Kim, and Y.~Song, ``Encrypting
  controller using fully homomorphic encryption for security of cyber-physical
  systems,'' \emph{IFAC-PapersOnLine}, vol.~49, no.~22, pp. 175--180, 2016.

\bibitem{schluter2023brief}
N.~Schl{\"u}ter, P.~Binfet, and M.~S. Darup, ``A brief survey on encrypted
  control: From the first to the second generation and beyond,'' \emph{Annual
  Reviews in Control}, p. 100913, 2023.

\bibitem{darup2021encrypted}
M.~S. Darup, A.~B. Alexandru, D.~E. Quevedo, and G.~J. Pappas, ``Encrypted
  control for networked systems: An illustrative introduction and current
  challenges,'' \emph{IEEE Control Systems Magazine}, vol.~41, no.~3, pp.
  58--78, 2021.

\bibitem{naseri2021securing}
A.~M. Naseri, W.~Lucia, M.~Mannan, and A.~Youssef, ``On securing cloud-hosted
  cyber-physical systems using trusted execution environments,'' in \emph{IEEE
  Int. Conf. on Autonomous Systems}, 2021, pp. 1--5.

\bibitem{costello2015geppetto}
C.~Costello, C.~Fournet, J.~Howell, M.~Kohlweiss, B.~Kreuter, M.~Naehrig,
  B.~Parno, and S.~Zahur, ``Geppetto: Versatile verifiable computation,'' in
  \emph{IEEE Symposium on Security and Privacy}.\hskip 1em plus 0.5em minus
  0.4em\relax IEEE, 2015, pp. 253--270.

\bibitem{parno2016pinocchio}
B.~Parno, J.~Howell, C.~Gentry, and M.~Raykova, ``Pinocchio: Nearly practical
  verifiable computation,'' \emph{Communications of the ACM}, vol.~59, no.~2,
  pp. 103--112, 2016.

\bibitem{bunz2018bulletproofs}
B.~B{\"u}nz, J.~Bootle, D.~Boneh, A.~Poelstra, P.~Wuille, and G.~Maxwell,
  ``Bulletproofs: Short proofs for confidential transactions and more,'' in
  \emph{IEEE Symposium on Security and Privacy}.\hskip 1em plus 0.5em minus
  0.4em\relax IEEE, 2018, pp. 315--334.

\bibitem{fiore2014efficiently}
D.~Fiore, R.~Gennaro, and V.~Pastro, ``Efficiently verifiable computation on
  encrypted data,'' in \emph{ACM SIGSAC Conf. on Computer and Communications
  Security}, 2014, pp. 844--855.

\bibitem{fiore2020boosting}
D.~Fiore, A.~Nitulescu, and D.~Pointcheval, ``Boosting verifiable computation
  on encrypted data,'' in \emph{Public-Key Cryptography: IACR International
  Conference on Practice and Theory of Public-Key Cryptography}.\hskip 1em plus
  0.5em minus 0.4em\relax Springer, 2020, pp. 124--154.

\bibitem{cheon2020authenticated}
J.~H. Cheon, D.~Kim, J.~Kim, S.~Lee, and H.~Shim, ``Authenticated computation
  of control signal from dynamic controllers,'' in \emph{IEEE Conf. on Decision
  and Control}.\hskip 1em plus 0.5em minus 0.4em\relax IEEE, 2020, pp.
  3249--3254.

\bibitem{mahfouzi2021secure}
R.~Mahfouzi, A.~Aminifar, S.~Samii, P.~Eles, and Z.~Peng, ``Secure cloud
  control using verifiable computation,'' in \emph{IEEE Int. Conference on
  Omni-Layer Intelligent Systems}.\hskip 1em plus 0.5em minus 0.4em\relax IEEE,
  2021, pp. 1--6.

\bibitem{Crépeau2005}
\BIBentryALTinterwordspacing
C.~Cr{\'e}peau, \emph{Cut-and-choose protocol}.\hskip 1em plus 0.5em minus
  0.4em\relax Boston, MA: Springer US, 2005, pp. 123--124. [Online]. Available:
  \url{https://doi.org/10.1007/0-387-23483-7_92}
\BIBentrySTDinterwordspacing

\bibitem{chaum2008scantegrity}
D.~Chaum and et~al., ``Scantegrity ii: End-to-end verifiability for optical
  scan election systems using invisible ink confirmation codes,'' \emph{EVT},
  vol.~8, no.~1, p.~13, 2008.

\bibitem{smart2016cryptography}
N.~P. Smart, \emph{Cryptography made simple}.\hskip 1em plus 0.5em minus
  0.4em\relax Springer, 2016.

\bibitem{escudero2023safety}
C.~Escudero, C.~Murguia, P.~Massioni, and E.~Zama{\"\i}, ``Safety-preserving
  filters against stealthy sensor and actuator attacks,'' in \emph{IEEE Conf.
  on Decision and Control (CDC)}.\hskip 1em plus 0.5em minus 0.4em\relax IEEE,
  2023, pp. 5097--5104.

\bibitem{tiriolo2022receding}
C.~Tiriolo, G.~Franz{\`e}, and W.~Lucia, ``A receding horizon trajectory
  tracking strategy for input-constrained differential-drive robots via
  feedback linearization,'' \emph{IEEE Trans. on Control Systems Technology},
  vol.~31, no.~3, pp. 1460--1467, 2022.

\bibitem{paillier1999public}
P.~Paillier, ``Public-key cryptosystems based on composite degree residuosity
  classes,'' in \emph{Int. Conf. on the Theory and Applications of
  Cryptographic Techniques}.\hskip 1em plus 0.5em minus 0.4em\relax Springer,
  1999, pp. 223--238.

\end{thebibliography}

\end{document}